\documentclass[%
 aip,
 jmp,%
 amsmath,amssymb,
 superscriptaddress,
 nofootinbib,
 reprint,%
]{revtex4}

\usepackage{amsmath,amsfonts,amssymb,color,epsfig,graphics,graphicx,latexsym,revsymb,theorem,url,verbatim}

\newtheorem{definition}{Definition}
\newtheorem{proposition}[definition]{Proposition}
\newtheorem{lemma}[definition]{Lemma}

\newtheorem{theorem}[definition]{Theorem}
\newtheorem{corollary}[definition]{Corollary}
\newtheorem{conjecture}[definition]{Conjecture}

\newtheorem{remark}[definition]{Remark}
\newtheorem{example}[definition]{Example}
\newtheorem{question}[definition]{Question}
\DeclareMathOperator{\Tr}{Tr}
\def\squareforqed{\hbox{\rlap{$\sqcap$}$\sqcup$}}
\def\qed{\ifmmode\squareforqed\else{\unskip\nobreak\hfil
\penalty50\hskip1em\null\nobreak\hfil\squareforqed
\parfillskip=0pt\finalhyphendemerits=0\endgraf}\fi}
\def\endenv{\ifmmode\;\else{\unskip\nobreak\hfil
\penalty50\hskip1em\null\nobreak\hfil\;
\parfillskip=0pt\finalhyphendemerits=0\endgraf}\fi}
\newenvironment{proof}{\noindent \textbf{{Proof.~} }}{\qed}
\def\Dbar{\leavevmode\lower.6ex\hbox to 0pt
{\hskip-.23ex\accent"16\hss}D}
\makeatletter
\def\url@leostyle{%
  \@ifundefined{selectfont}{\def\UrlFont{\sf}}{\def\UrlFont{\small\ttfamily}}}
\makeatother
\def\bcj{\begin{conjecture}}
\def\ecj{\end{conjecture}}
\def\bcr{\begin{corollary}}
\def\ecr{\end{corollary}}
\def\bd{\begin{definition}}
\def\ed{\end{definition}}
\def\bea{\begin{eqnarray}}
\def\eea{\end{eqnarray}}
\def\bem{\begin{enumerate}}
\def\eem{\end{enumerate}}
\def\bex{\begin{example}}
\def\eex{\end{example}}
\def\bim{\begin{itemize}}
\def\eim{\end{itemize}}
\def\bl{\begin{lemma}}
\def\el{\end{lemma}}
\def\bpf{\begin{proof}}
\def\epf{\end{proof}}
\def\bpp{\begin{proposition}}
\def\epp{\end{proposition}}
\def\bqu{\begin{question}}
\def\equ{\end{question}}
\def\br{\begin{remark}}
\def\er{\end{remark}}
\def\bt{\begin{theorem}}
\def\et{\end{theorem}}

\def\btb{\begin{tabular}}
\def\etb{\end{tabular}}

\newcommand{\nc}{\newcommand}

\def\e{\epsilon}

 \nc{\bA}{{\bf A}} \nc{\bB}{{\bf B}} \nc{\bC}{{\bf C}}
 \nc{\bD}{{\bf D}} \nc{\bE}{{\bf E}} \nc{\bF}{{\bf F}}
 \nc{\bG}{{\bf G}} \nc{\bH}{{\bf H}} \nc{\bI}{{\bf I}}
 \nc{\bJ}{{\bf J}} \nc{\bK}{{\bf K}} \nc{\bL}{{\bf L}}
 \nc{\bM}{{\bf M}} \nc{\bN}{{\bf N}} \nc{\bO}{{\bf O}}
 \nc{\bP}{{\bf P}} \nc{\bQ}{{\bf Q}} \nc{\bR}{{\bf R}}
 \nc{\bS}{{\bf S}} \nc{\bT}{{\bf T}} \nc{\bU}{{\bf U}}
 \nc{\bV}{{\bf V}} \nc{\bW}{{\bf W}} \nc{\bX}{{\bf X}}
 \nc{\bZ}{{\bf Z}}

\nc{\cA}{{\cal A}} \nc{\cB}{{\cal B}} \nc{\cC}{{\cal C}}
\nc{\cD}{{\cal D}} \nc{\cE}{{\cal E}} \nc{\cF}{{\cal F}}
\nc{\cG}{{\cal G}} \nc{\cH}{{\cal H}} \nc{\cI}{{\cal I}}
\nc{\cJ}{{\cal J}} \nc{\cK}{{\cal K}} \nc{\cL}{{\cal L}}
\nc{\cM}{{\cal M}} \nc{\cN}{{\cal N}} \nc{\cO}{{\cal O}}
\nc{\cP}{{\cal P}} \nc{\cQ}{{\cal Q}} \nc{\cR}{{\cal R}}
\nc{\cS}{{\cal S}} \nc{\cT}{{\cal T}} \nc{\cU}{{\cal U}}
\nc{\cV}{{\cal V}} \nc{\cW}{{\cal W}} \nc{\cX}{{\cal X}}
\nc{\cZ}{{\cal Z}}

\nc{\hA}{{\hat{A}}} \nc{\hB}{{\hat{B}}} \nc{\hC}{{\hat{C}}}
\nc{\hD}{{\hat{D}}} \nc{\hE}{{\hat{E}}} \nc{\hF}{{\hat{F}}}
\nc{\hG}{{\hat{G}}} \nc{\hH}{{\hat{H}}} \nc{\hI}{{\hat{I}}}
\nc{\hJ}{{\hat{J}}} \nc{\hK}{{\hat{K}}} \nc{\hL}{{\hat{L}}}
\nc{\hM}{{\hat{M}}} \nc{\hN}{{\hat{N}}} \nc{\hO}{{\hat{O}}}
\nc{\hP}{{\hat{P}}} \nc{\hR}{{\hat{R}}} \nc{\hS}{{\hat{S}}}
\nc{\hT}{{\hat{T}}} \nc{\hU}{{\hat{U}}} \nc{\hV}{{\hat{V}}}
\nc{\hW}{{\hat{W}}} \nc{\hX}{{\hat{X}}} \nc{\hZ}{{\hat{Z}}}

\newcommand{\bra}[1]{\langle#1|}
\newcommand{\ket}[1]{|#1\rangle}

\def\Dbar{\leavevmode\lower.6ex\hbox to 0pt
{\hskip-.23ex\accent"16\hss}D}

\begin{document}

\title{Local density matrices of many-body states\\ in the constant weight subspaces}

\author{Jianxin Chen}
\affiliation{Joint Center for Quantum Information and Computer Science, University of Maryland, College Park, Maryland, USA}
\author{Muxin Han}%
\affiliation{Department of Physics, Florida Atlantic University, FL 33431, USA}%
\affiliation{Institut f\"ur Quantengravitation, Universit\"at Erlangen-N\"urnberg, Staudtstr. 7/B2, 91058 Erlangen, Germany}
\author{Youning Li}\thanks{Author to whom correspondence should be addressed: liyn03@gmail.com}%
\affiliation{Department of Physics, Tsinghua University, Beijing, People's Republic of China}%
\affiliation{Collaborative Innovation Center of Quantum Matter, Beijing 100190, People's Republic of China}%
\author{Bei Zeng}
\affiliation{Department of Mathematics \& Statistics, University of Guelph, Guelph, Ontario, Canada}%
\affiliation{Institute for Quantum Computing, University of Waterloo, Waterloo, Ontario, Canada}
\author{Jie Zhou}
\affiliation{Perimeter Institute for Theoretical Physics, Waterloo, Ontario, Canada}

\date{\today}

\begin{abstract}
Let $V=\bigotimes_{k=1}^{N} V_{k}$ be an $N$-particle Hilbert space, whose individual single-particle space is the one with spin $j$ and dimension $d=2j+1$.  Let $V_{(w)}$ be the subspace of $V$ with constant weight $w$,
consisting of vectors whose total spins are $w$.
We show that the combinatorial properties of the constant weight condition impose strong constraints on the reduced density matrices for any vector $\ket{\psi}$ in the constant weight subspace $V_{(w)}$, which limit the possibility of the entanglement structures of $\ket{\psi}$. Our results find applications in the overlapping quantum marginal problem, quantum error-correcting codes, and the spin-network structures in quantum gravity.
\end{abstract}

\maketitle
\renewcommand\theequation{\arabic{section}.\arabic{equation}}
\setcounter{tocdepth}{4}
\makeatletter
\@addtoreset{equation}{section}
\makeatother

\section{Introduction}
\label{sec:Intro}

Consider a system of $N$ particles, each with spin $j$ and dimension $d=2j+1$. The Hilbert space of the system is then $V=\bigotimes_{k=1}^{N} V_{k}$, where all of the $V_k$'s are identical.
For any vector ${\psi}\in{V}$, its entanglement structure can be analyzed by
the reduced density matrices (RDMs) of $\psi\otimes\psi^*\in V\otimes V^*$~\cite{toth2007detection,horodecki2009quantum}. There are, however, always restrictions on these RDMs, given by, e.g. the principle of entanglement monogamy~\cite{toth2007detection,horodecki2009quantum}.
In a more general setting, the quantum marginal problem considers the consistency of a set of given local density matrices. This problem turns out to be a hard problem even with the existence of quantum computers.
On the other hand, generic states are always highly entangled, in the large-$j$ limit~\cite{hayden2006aspects}.

Besides those general analyses, there are also physical considerations that may restrict the form of ${\psi}$, hence the entanglement structures of ${\psi}$. For instance, ground states of local Hamiltonians would satisfy the entanglement area law, hence may be well-approximated by the tensor network representation~\cite{zeng2015quantum,eisert2010colloquium,verstraete2006criticality}. States with special symmetry are also discussed such as the Dicke states and their generalizations~\cite{toth2007detection,wieczorek2009experimental}. Stabilizer/graph states are considered in the scenario of quantum error correction and one-way quantum computing~\cite{calderbank1998quantum,gottesman1997stabilizer,raussendorf2001one,hein2004multiparty}. Very recently, states that may be represented by (restricted) Boltzman machine are considered to apply machine learning techniques to study many-body ground states~\cite{carleo2017solving}.

In this work, we consider the restriction to constant weight subspaces.
A state $\psi$ in $V$ possesses a constant weight $w$, if it lies in the subspace which has
an orthonormal basis $\{\ket{m_1,\ldots,m_N}\}$ satisfying $\sum_i m_i=w$. In the spin language, the state has a fixed $J_z$-component of the total spin. These subspaces arise naturally as decoherence free subspaces under collective dephasing of the system~\cite{lidar1998decoherence} -- that is, since each qubit gets a phase factor that only depends on its own weight, any constant weight state obtains a global phase for the collective dephasing. In this sense, a constant weight state is invariant under the collective dephasing noise, hence is decoherence free. Also, constant weight states are discussed in many other contexts, such as the atomic Dicke states and its generalizations. When $w=0$, this subspace contains the invariant subspace of zero angular momentum, which is widely discussed in loop quantum gravity.
Despite that the constant-weight condition arise naturally in these many places, the entanglement structure of these spaces has not been studied systematically.

We discuss the properties of the RDMs of constant weight states in a very general setting.
We show that there exist linear conditions between the elements of RDMs, for any $j$, $N$ and $w$, which can be written down explicitly. Our key idea is that the combinatorial properties given by the constant weight constraint, which is mathematically a partition of $w$, lead to such linear structures of the reduced density matrices. These conditions could find many applications. For instance, it implies that there is no perfect tensor in a constant weight subspace, for any $j$ when $N\geq 4$, which is a concept that has recently received attention from understanding quantum gravity from the quantum information viewpoint~\cite{pastawski2015holographic,li2017invariant}. Also, given the intimate connections between perfect tensors and quantum error-correcting codes, our results give restrictions on the achievable distance on constant-weight quantum codes. In practice, our conditions can also be used as good certificates for decoherence-free subspaces.

We organize our paper as follows. In Sec.~\ref{sec:pre}, we define our notions and provide background information on constant weight subspaces in the $\mathrm{SU}(2)$ case. In Sec.~\ref{sec:combinatorics} we discuss the combinatorial structure of the constant weight condition that leads to our main theorem on linear relations of RDMs.
In Sec.~\ref{sec:stru}, we discuss further the relationships between these linear structures.
In Sec.~\ref{sec:app}, we discuss the application of our main results to
the quantum marginal problem
and the
nonexistence of perfect tensors. In Sec.~\ref{sec:sun}, we discuss the generalization to the $\mathrm{SU}(n)$ case. Finally, in Sec.~\ref{sec:gen}, we discuss the generalization for relaxing the constant weight condition by introducing the notion of frequency matrix.

\section{Preliminaries}
\label{sec:pre}

According to the standard representation theory \cite{fulton2013representation}, the study of representations of the Lie group $\mathrm{SU}(2)$
is essentially equivalent to that of the Lie algebra $\mathfrak{su}(2)$.
In this work, we shall use the language of the latter for simplicity.\\

The Lie algebra $\mathfrak{su}(2)$ is generated by the Chevalley-Serre basis $\{H,X,Y\}$,
whose Lie algebra structure is given by
\begin{equation}
[H,X]=2X\,,
\quad
[H,Y]=-2Y\,,
\quad
[X,Y]=H\,.
\end{equation}
They act on the standard representation $\mathbb{C}^2$ of $\mathfrak{su}(2)$  by multiplication by the following matrices:
\begin{equation}
H=
\begin{pmatrix}
1 & 0\\
0&-1
\end{pmatrix}\,,
\quad
X=
\begin{pmatrix}
0& 1\\
0&0
\end{pmatrix}\,,
\quad
Y=
\begin{pmatrix}
0& 0\\
1&0
\end{pmatrix}\,.
\end{equation}
These matrices are related to the Pauli matrices
\begin{equation}
J_z=
\begin{pmatrix}
1 & 0\\
0&-1
\end{pmatrix}\,,
\quad
J_x=
\begin{pmatrix}
0& 1\\
1&0
\end{pmatrix}\,,
\quad
J_y=
\begin{pmatrix}
0& -i\\
i &0
\end{pmatrix}\,
\end{equation}
by
\begin{equation}
X= \frac{1}{2}(J_{x}+iJ_{y})=J_{+}\,,\quad
Y=\frac{1}{2} (J_{x}-i J_{y})=J_{-}\,,
\quad
H= J_{z}=J_{3}\,.
\end{equation}

The finite dimensional irreducible representations of $\mathfrak{su}(2)$ are classified by the dimension $d\in \mathbb{Z}_{+}$.
One has
for each $d\in \mathbb{Z}_{+}$ an irreducible representation
$\mathrm{Sym}^{\otimes 2j}\mathbb{C}^{2}$ of dimension $d=2j+1$,
where $\mathbb{C}^{2}$ is the standard representation.\\

Each representation $W$ of dimension $D$, not necessarily irreducible,
can be decomposed into a direct sum of irreducible representations.
Moreover, there exists a Hermitian metric $\langle-,- \rangle$ on $W$
such that the decomposition is orthogonal.
We shall denote the
dual of $W$ with respect to the Hermitian metric by $W^{*}$,
and the dual
of a vector $v\in W$
by $v^{*}$.

According to the weight decomposition of irreducible representations, one can
then find an orthonormal basis of $W$
\begin{equation}
\mathcal{B}=\{e_{1},e_{2},\cdots, e_{D}\}
\end{equation}
whose
weights, the eigenvalues under the action of $H$, are
\begin{equation}
\alpha_{1},\alpha_{2},\cdots, \alpha_{D}\,,
\end{equation}
respectively.
Note that here the weight is 2 times the usual notion of spin.
See Figure. \ref{figureirredecomposition} for an illustration.\\

\begin{figure}[h]
  \renewcommand{\arraystretch}{1.2}
\begin{displaymath}
\includegraphics[scale=0.8]{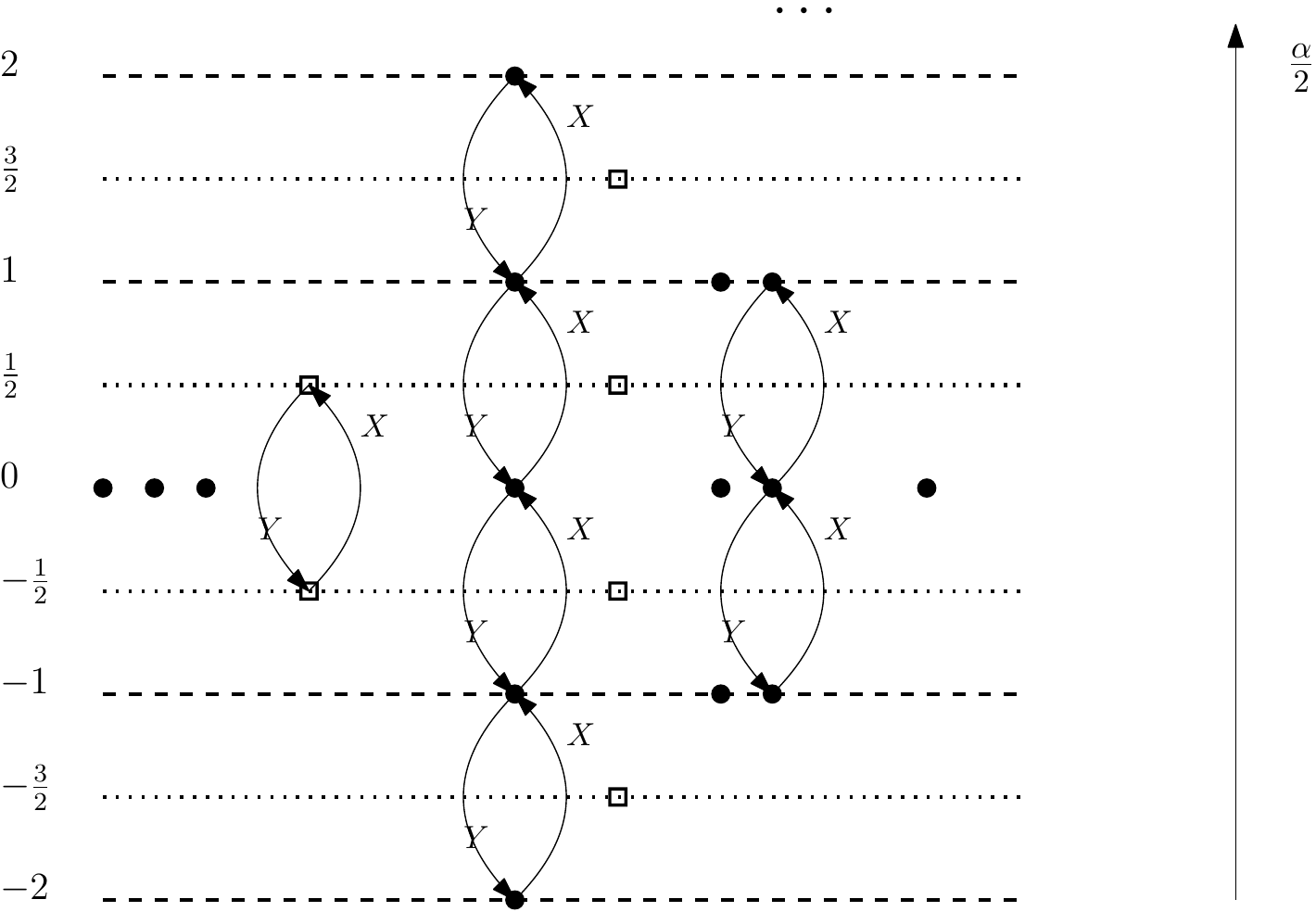}
\end{displaymath}
  \caption[irredecomposition]{The decomposition of a representation $W$ into direct sum of irreducible representations.
Each disk/square represents a one-dimensional eigenspace of $H$ with integer/half integer spins, respectively.
The vector spaces aligned in the vertical direction, which are connected by actions of $X$ and $Y$, constitute an irreducible representation.}
  \label{figureirredecomposition}
\end{figure}

We label each of the vectors in the basis $\mathcal{B}$ for $W$ by its sub-index $r\in \{1,2,\cdots, D\}$, and vice versa.
We shall adapt this convenient convention
throughout this work.\\


We consider in this work the tensor product
\begin{equation}
V=\bigotimes_{k=1}^{N} V_{k}\,,
\end{equation}
where all of the components $V_{k}$ are identical  to some given representation, not necessarily irreducible, say $W$. In our following discussion, we shall consider the non-trivial case $N\geq2$.\\

A basis of $V$ is then provided by $\mathcal{B}^{\otimes N}$, whose elements are indexed by the multi-indices
$I=(i_{1},i_{2},\cdots, i_{N})$ corresponding to the vector
\begin{equation}\label{eqnbasisforV}
e_{I}:=e_{i_{1}}\otimes e_{i_{2}}\cdots \otimes e_{i_{N}}\,.
\end{equation}
The weight of this vector is easily seen to be
\begin{equation}
\mathrm{weight}(I):=\alpha_{i_{1}}+\alpha_{i_2}+\cdots \alpha_{i_{N}}\,.
\end{equation}
Any vector $\psi$ in $V$ is then represented by
\begin{equation}
\psi=\sum_{e_I\in \mathcal{B}^{\otimes N}} a_{I}e_{I}\,,
\quad a_{I}\in \mathbb{C}\,.
\end{equation}

We now discuss the notion of partial trace. Choose a subset of components $\Lambda \subseteq \{1,2,\cdots N\}$ for $ V$,  and define
\begin{equation}
V_{\Lambda}=\bigotimes_{k\in \Lambda} V_{k}\,,
\quad
V_{\Lambda^{c}}=\bigotimes_{k\notin \Lambda} V_{k}\,.
\end{equation}
The identity operator $\mathbb{I}_{V_{\Lambda}}\in \mathrm{End}(V_{\Lambda}) $
is equivalently represented by a unique tensor $\Delta_{V_{\Lambda}}\in V_{\Lambda} \otimes \check{V}_{\Lambda}$, or alternatively a unique tensor $\check{\Delta}_{V_{\Lambda}}\in \check{V}_{\Lambda} \otimes V _{\Lambda}$, where the notation $\check{}$ means the linear dual in the category of vector spaces.
The explicit formula for $\Delta_{V_{\Lambda}}$ is displayed in \eqref{eqntensorformofidentity} below.

The Hermitian metric gives an identification between the Hermitian dual $V^{*}$ and the linear dual $\check{V}$.
This identification will be assumed frequently in this work.
By this identification, the element $\psi\otimes \psi^{*}\in V\otimes V^{*}$ then determines an element in $\mathrm{End} (V)=V \otimes \check{V}$.
Thus one can contract it with the tensor
$\Delta_{V_{\Lambda}}$. This is equivalent to the following pairing using the Hermitian metric
\begin{equation}\label{eqndefofpartialtrace}
\langle \Delta_{V_{\Lambda}} ,  \psi\otimes \psi^{*} \rangle\in \mathrm{End} (V_{\Lambda^{c}})\,.
\end{equation}
As priori, $\langle \Delta_{V_{\Lambda}} ,  \psi\otimes \psi^{*} \rangle$ is only an element in $V_{\Lambda^c} \otimes V_{\Lambda^c}^*$, but in \eqref{eqndefofpartialtrace}, we have used such identification to identify it with an element in $\mathrm{End} (V_{\Lambda^{c}})=V_{\Lambda^{c}} \otimes \check{V_{\Lambda^{c}}}$.

\begin{definition}[Partial trace]
The partial trace of $\psi\otimes \psi^{*}$ over the vector space $V_{\Lambda}$ is defined to be
\begin{equation}
\mathrm{Tr}_{V_{\Lambda}}( \psi\otimes \psi^{*} ):=\langle \Delta_{V_{\Lambda}} ,  \psi\otimes \psi^{*} \rangle\,.
\end{equation}
\end{definition}

The above definition can be applied to a general element in $V\otimes V^{*}$ which is not necessarily of the form $\psi\otimes \psi^{*}$.\\

Writing $I=(L; K)$, where
$K$ runs over the index set for
the orthonormal basis
$\mathcal{B}^{\otimes |\Lambda|}=\{e_{K}\}$ of $V_{\Lambda}$ and
$L$ over that for  $V_{\Lambda^{c}}$, we then have
\begin{equation}\label{eqntensorformofidentity}
 \Delta_{V_{\Lambda}} =\sum_{e_{K}\in \mathcal{B}^{\otimes |\Lambda|}} e_{K}\otimes e_{K}^{*}\,.
\end{equation}
The notation $|\Lambda|$ stands for the cardinality of the index set $\Lambda$.

Hence (hereafter $\Theta_D :=\{1,2,\cdots, D\}$)
\begin{equation}
\mathrm{Tr}_{V_{\Lambda}}( \psi\otimes \psi^{*} )=
\sum_{L, L'}
\sum_{K\in \Theta_D^{\otimes |\Lambda|} }a_{(L;K)} a^{*}_{(L';K)} e_{L}\otimes e_{L'}^{*}\in \mathrm{End} (V_{\Lambda^{c}})\,.
\end{equation}

\begin{example}
Consider the $D=2, N=4$ case.
Given a state $\ket{\psi}=\sum_{i_1,i_2,i_3,i_4} a_{i_1,i_2,i_3,i_4} \ket{i_1,i_2,i_3,i_4}$, taking partial trace over $V_1\otimes V_4$ can be calculated by
\begin{equation}
\Tr_{V_1\otimes V_4}\left(\ket{\psi}\bra{\psi}\right)
=\sum_{i_2,i_3,i_2',i_3'=1,2}\ket{i_2i_3}\bra{i_2'i_3'}
\sum_{i_1,i_4=1,2}a_{i_1i_2i_3i_4}a^{*}_{i_1i_2'i_3'i_4}\,.
\end{equation}
\end{example}

With respect to the basis we have chosen, $\mathrm{Tr}_{V_{\Lambda}}( \psi\otimes \psi^{*} )$
is naturally represented by its entries
 \begin{equation}\label{eqnentriesofpartialtrace}
\left(\mathrm{Tr}_{V_{\Lambda}}( \psi\otimes \psi^{*} )\right)_{L, L'}=
\sum_{K\in \Theta_D^{\otimes |\Lambda|} }a_{(L;K)} a^{*}_{(L';K)} \,,
\quad
L, L'\in \Theta_D^{\otimes|\Lambda^{c}|}\,.
\end{equation}
The tensor $\psi\otimes \psi^{*}\in\mathrm{End}(V_{\Lambda^{c}}\otimes V_{\Lambda})$ has rank one, hence
the dimension of its kernel is
$\mathrm{dim} (V_{\Lambda^{c}}\otimes V_{\Lambda}) -1$.
Taking the partial trace over $V_{\Lambda}$ would
at most increase the rank of the resulting partial trace by $ \mathrm{dim}  V_{\Lambda}$:
forgetting about the component $V_{\Lambda}$ in $V_{\Lambda^{c}}\otimes V_{\Lambda}$ would at most decrease the dimension of the kernel by $ \mathrm{dim}  V_{\Lambda}$.
Therefore, the rank of $\mathrm{Tr}_{V_{\Lambda}}( \psi\otimes \psi^{*} )$ has an upper bound
\begin{equation}
\mathrm{rank}(\mathrm{Tr}_{V_{\Lambda}}( \psi\otimes \psi^{*} ))\leq 1+ \mathrm{dim} V_{\Lambda}\,.
\end{equation}
In order that the partial trace, as an element in $V_{\Lambda^{c}}\otimes V_{\Lambda^{c}}^{*}$, has full rank, the following condition has
to be met
\begin{equation}
\mathrm{dim} V_{\Lambda^{c}}\leq  1+ \mathrm{dim} V_{\Lambda} \,.
\end{equation}
In the present case, all of the components are isomorphic representations.
Hence the above condition reduces to
\begin{equation}\label{eqnqualitativeanalysis}
|\Lambda|\geq [{N\over 2}]\,.
\end{equation}
Intuitively, one must contract enough components in order for the resulting partial trace
to have possibly maximal rank.

\section{Combinatorics in partial trace on the constant weight subspace}
\label{sec:main}

We shall discuss in this section some combinatorial
structure of partial trace and of the constant weight subspace,
based on which
we shall discuss some applications in Section~\ref{sec:app} .\\

Among the entries in the partial trace \eqref{eqnentriesofpartialtrace},
of particular interest are the diagonal ones
 \begin{equation}\label{eqndfnrhodiagonal}
\rho_{L}^{{\Lambda^{c}}}
:=
\sum_{K\in \Theta_D^{\otimes |\Lambda|}}|a_{(L;K)} |^2 \,,
\quad
L\in \Theta_D^{\otimes|\Lambda^{c}|}\,.
\end{equation}

We fix $\Lambda^{c}=\{1,2,\cdots, M, M+1\}$ for some $0\leq M\leq N-1$.
Writing the index set $L$ as $(I_{0}; i_{M+1})$, where $I_0=(i_1, i_2,\cdots,i_M)$ is a multi-index
and, $1 \leq i_k \leq D, k=1,2,\cdots M, M+1$. Then the diagonal pieces of the partial trace
over $V_{\Lambda}$ are represented by the entries
\begin{equation}\label{eqndiagonalpieces}
\rho_{(I_{0};i_{M+1} )}^{{\Lambda^{c}}}
=
\sum_{K\in \Theta_D^{\otimes |\Lambda|}} |a_{(I_0;i_{M+1};K)} |^2 \,,
\quad
\forall\, (I_{0};i_{M+1})\in \Theta_D^{\otimes|\Lambda^{c}|}\,.
\end{equation}
An element $K$ takes the form $(i_{M+2}, \cdots, i_{N})$.
By moving the position of $i_{M+1}$ from $M+1$ to a position valued in the set $\{M+1,M+2,\cdots, N\}$ which symbolically is denoted by $*$,
we get similarly the quantities.
To be more precise,
denote
\begin{equation*}
\Lambda^{c}_{*}=\{1,2,\cdots, M, *\}\,,
\quad
\Lambda_{*}= \{1,2,\cdots N\}-\Lambda^{c}_{*}\,,
\end{equation*}
then we have
\begin{equation}
\rho_{(I_{0}; i_*)}^{\Lambda^{c}_{*}=\{1,2,\cdots, M, *\}}=
\sum_{(i_{M+1},\cdots,\widehat{ i_{*}}, \cdots, i_{N})  \in \Theta_D^{\otimes |\Lambda_{*}|}} |a_{(I_0; i_{M+1},\cdots, i_{*}, \cdots, i_{N}) } |^2 \,,
\quad
\forall\, (I_{0};i_{*})\in \Theta_D^{\otimes|\Lambda_{*}^{c}|}\,.
\end{equation}
Here the notation $\widehat{ i_{*}}$
means the index $ i_{*}$ is omitted.
For simplicity, we denote
$K_{*}=(i_{M+1},\cdots,\widehat{ i_{*}}, \cdots, i_{N})$ and $K'=( i_{M+1},\cdots, i_{*}, \cdots, i_{N})$.
Since the sub-index $k$ of $i_{k}$ already contains the information $i_{k}\in V_{k}$, we can
denote for convenience $K'=(i_{*}; K_{*})$.

Later we shall consider
 \begin{equation}\label{eqnsummationoverrexpanded}
\sum_{*\in  \{M+1,M+2,\cdots, N\}  }
\sum_{\substack{i_*=1\\e_{i_*}\in V_{*} }}^{D}
\rho_{(I_{0};i_* )}^{\{1,2,\cdots , M, *\} }
=
\sum_{*\in  \{M+1,M+2,\cdots, N\}}
\sum_{\substack{i_*=1\\e_{i_*}\in V_{*} }}^{D}
\sum_{K_*  \in \Theta_D^{\otimes |\Lambda_{*}|}} |a_{(I_0; i_*; K_*)} |^2
\,,
\end{equation}
where $e_{i_*} \in V_*$
indicates that the index $i_*$ labels different basis vectors in $V_{*}$.

\subsection{Combinatorial identity in the constant weight subspace}
\label{sec:combinatorics}

After restricting to the constant weight subspace $V_{(w)}$ of $V$ consisting of vectors of weight $w$,
an index $K'$ that appear in the sum \eqref{eqnsummationoverrexpanded}
corresponds to a $(N-M)$-tuple $x=(x_{M+1}, \cdots ,x_{*}, \cdots, x_N)$, where $x_{k}$
is the weight of the vector labeled by $i_{k}$ for $k=M+1,\cdots, *, \cdots , N$.
The constant weight condition is translated to the condition that $x$ is a solution to the equation
\begin{equation}\label{eqnspaceoforderedindices}
\sum_{k=M+1}^{N}x_{k}=-\mathrm{weight}(I_{0})+w\,,
\quad
x_{k}\in \{\alpha_{1},\alpha_{2},\cdots ,\alpha_{D}\}\,,
\end{equation}
where the subindex $k$ of $x_k$ labels different components $V_k$, and the subindex $r$ of $\alpha_r$ labels different basis vectors in each $V_k$.
We denote the latter set of solutions by $\mathcal{X}$.

Note that different elements in $\mathcal{X}$ can correspond to the same partition:
permuting an $(N-M)$-tuple can give a different $(N-M)$-tuple but
they correspond to the same partition.
After modulo the action by the symmetry group $\mathfrak{S}_{N-M}$, the set $\mathcal{X}/\mathfrak{S}_{N-M}$ of cosets
is then in one-to-one correspondence with the set of partitions
(with all elements in a partition valued in $\{\alpha_{1},\alpha_{2},\cdots ,\alpha_{D}\}$) of the following function $S$ of weight($I_0$) and $w$
\begin{equation}
S=-\mathrm{weight}(I_{0})+w\,.
\end{equation}
More precisely, any element in the set $\mathcal{X}/\mathfrak{S}_{N-M}$, denoted by $[x]$,
is given by a partition
\begin{equation}\label{eqnspaceofpartitions}
\alpha_{1}\cdot n_{1}([x])+\alpha_{2}\cdot n_{2}([x])+\cdots + \alpha_{D} n_{D}([x])\,.
\end{equation}
Here $n_{r}([x]), r=1,2,\cdots, D$ is the frequency of $\alpha_{r}$ appearing
in the partition, which is independent of the choice of the representative of the coset $[x]$.
They are subject to the conditions that
\begin{equation}\label{eqnconstantweightcombinatorics}
\sum_{r=1}^{D}\alpha_{r}n_{r}([x])=S\,,
\quad
\sum_{r=1}^{D} n_{r}([x])=N-M\,.
\end{equation}

Therefore, we get the following result.

\begin{lemma}\label{lemmacombinatoricsofpartition}
There exists a nonzero solution $\{b_{r}\}_{r=1}^{D}$ to the equation
\begin{equation}\label{eqncombinatoricsofpartition}
\sum_{r=1}^{D}b_{r}n_{r}([x])=0\,,\quad
\forall \, [x]\in \mathcal{X}/\mathfrak{S}_{N-M}\,.
\end{equation}
An explicit one is given by
\begin{equation}\label{eqnsoltobr}
b_{r}=\alpha_{r}-{S\over N-M}\,,
\quad
r=1,2,\cdots, D\,.
\end{equation}
\end{lemma}

\begin{remark}
A natural question is about the uniqueness of the solution.
This will be addressed in Section \ref{sec:gen} below using the notion of frequency matrix.
\end{remark}

\begin{example}\label{exfrequencymatrix}
Take $N=5, j={1}$.
Then $D=3$ and $\alpha_{r}/2\in \{-1,0,1\}$ for $r\in \{1,2,3\}$.
Consider the case $M=1,w=0$.
Labeling the basis in $V_k$ by spin, which is half of the eigenvalue of $H$ acting on $V_k$, we then have the following
\begin{displaymath}
\begin{tabular}{ccccc}
         &  $\mathrm{weight}(I_{0})/2$  & $\{[x]\}$  & $(n_{r}([x])) $ & $b=(b_{r})^{t}$ \\
\hline
       & $-1$  &$\substack{ [1,0,0,0]\\
[1,0,-1,1] }$ &
$
\begin{pmatrix}
0&3 & 1\\
1&1 & 2
\end{pmatrix}
$
 &
$
\begin{pmatrix}
-{5/4}\\
-{1/4}\\
{3/4}
\end{pmatrix}
$
\\
\hline
     &  $0$  & $\substack{ [1,-1,1,-1]\\
[1,-1,0,0] \\
[0,0,0,0]}$   & $
\begin{pmatrix}
2&0 & 2\\
1&2 & 1\\
0&4 &0
\end{pmatrix}
$
 &
$
\begin{pmatrix}
-{1/4}\\
0\\
{1/ 4}
\end{pmatrix}
$  \\
\hline
     &  $1$  & $\substack{ [-1,0,-1,1]\\
[-1,0,0,0] }$   &
$
\begin{pmatrix}
2&1 & 1\\
1&3 & 0
\end{pmatrix}
$
 &
$
\begin{pmatrix}
-{3/ 4}\\
{1/ 4}\\
{5/4}
\end{pmatrix}
$
 \end{tabular}
\end{displaymath}

\end{example}

We can now prove the following theorem.

\begin{theorem}\label{thmeqnforpartialtraces}
Consider the constant weight subspace $V_{(w)}$ of $V$ consisting of vectors of weight $w$.
Fix an integer $1\leq M\leq N-1$ and an index $I_{0}\in \Theta_{D}^{\otimes M}$
such that the set $\mathcal{X}$ of solutions to
\eqref{eqnspaceoforderedindices} is non-empty.
Then there exists a nonzero solution $\{b_{r}\}_{r=1}^{D}$ to the equation
\begin{equation}\label{eqnforpartialtraces}
\sum_{r=1}^{D} b_{r}
\sum_{\substack{ *\in  \{M+1,M+2,\cdots, N\}  \\
e_r\in V_{*}}}\rho^{\{1,2,\cdots, M,*\}}_{(I_{0};r)}=0\,,
\end{equation}
where $e_r\in V_{*}$ indicates that $r$ labels different basis vectors in $V_*$.

\end{theorem}
\begin{proof}
Straightforward computation as in \eqref{eqnsummationoverrexpanded} shows that
\begin{eqnarray*}
&&\sum_{ *\in  \{M+1,M+2,\cdots, N\} }
\sum_{\substack{ i_*=1\\
e_{i_*}\in V_{*}}
}^{D} b_{i_*}
\rho^{\{1,2,\cdots, M,*\}}_{(I_{0};i_*)}
=
\sum_{*\in  \{M+1,M+2,\cdots, N\}}
\sum_{\substack{i_*=1\\e_{i_*}\in V_{*} }}^{D}b_{i_*}
\sum_{K_{*}  \in \Theta_D^{\otimes |\Lambda_{*}|}} |a_{(I_0; i_{*}; K_{*}) } |^2.
\end{eqnarray*}
After interchanging the order of summation on $i_*$ and $*$, one would obtain
\begin{eqnarray}\label{eqnintermediatesimplifcation}
\sum_{r=1}^{D}b_{r}
\sum_{\substack{*\in  \{M+1,M+2,\cdots, N\}\\ e_r\in V_{*}}}
\sum_{K_{*}  \in \Theta_D^{\otimes |\Lambda_{*}|}} |a_{(I_0; r; K_{*}) } |^2
=
\sum_{r=1}^{D}b_{r}
\sum_{\substack{*\in  \{M+1,M+2,\cdots, N\}\\ e_r\in V_{*}}}
\rho^{\{1,2,\cdots, M,*\}}_{(I_{0};r)}
\,.
\end{eqnarray}
Now we restrict ourselves to the constant weight subspace and hence replace the multi-index $(i_*;K_*)$
by a solution $x\in \mathcal{X}$ given in \eqref{eqnspaceoforderedindices}.
Then the LHS in \eqref{eqnintermediatesimplifcation}
gives
\begin{eqnarray}\label{eqnequivalentexpression}
\sum_{r=1}^{D}b_r \sum_{x \in \mathcal{X}} \sum_{\substack{ *\in \{M+1,M+2,\cdots,N\}  \\
x_*= \alpha_r} }
|a_{(I_{0}; x)}|^{2}\,.
\end{eqnarray}
Theorem \ref{thmeqnforpartialtraces} is then equivalent to the vanishing of
\eqref{eqnequivalentexpression} above.

We define $\mathrm{Supp}(x)$ to be the set of distinct values in the entries of $x$.
This gives a function on the set $\mathcal{X}$ of solutions.
It is obvious that it is invariant under the action by the group $\mathfrak{S}_{N-M}$ and hence descends to a function
on the set
$\mathcal{X}/\mathfrak{S}_{N-M}$ of partitions.
We then define
$\mathrm{Supp}([x])$ to be $\mathrm{Supp}(x)$ for any representation $x$ of $[x]$.
 Therefore,
\begin{eqnarray*}
&&
\sum_{r=1}^{D}b_r \sum_{x \in \mathcal{X}} \sum_{\substack{ *\in \{M+1,M+2,\cdots,N\}  \\
x_*= \alpha_r} }
|a_{(I_{0}; x)}|^{2}\\
&=&\sum_{x \in \mathcal{X}} \sum_{r=1}^{D}b_r  \sum_{\substack{ *\in \{M+1,M+2,\cdots,N\}  \\
x_*= \alpha_r} }
|a_{(I_{0}; x)}|^{2}\\
&=&
\sum_{x\in   \mathcal{X}}
\sum_{r:\, \alpha_r\in \mathrm{Supp}(x)}  b_r
n_{r}(x) |a_{(I_{0}; x)}|^{2}\\
&=&
\sum_{x\in  \mathcal{X}/\mathfrak{S}_{N-M} } \sum_{x\in   [x]}
\sum_{r: \,\alpha_r\in \mathrm{Supp}([x])}  b_{r}
n_{r}([x]) |a_{(I_{0}; x)}|^{2}\\
&=&\sum_{x\in  \mathcal{X}/\mathfrak{S}_{N-M} }  \sum_{ r: \,\alpha_r\in \mathrm{Supp}(x)}  b_{r}
n_{r}([x])  \sum_{x\in [x]} |a_{(I_{0}; x)}|^{2}\\
&=&\sum_{[x]\in  \mathcal{X}/\mathfrak{S}_{N-M}  }\left(  \sum_{ r: \,\alpha_r\in \mathrm{Supp}([x])}  b_{r}
n_{r}([x]) \right)\left(\sum_{\substack{x\in [x]} }|a_{(I_{0}; x)}|^{2}\right )
\,.
\end{eqnarray*}
For any partition $[x]$, if $\alpha_{r}\notin  \mathrm{Supp}([x])$, then
$n_{r}([x])=0$ automatically. It follows that
\begin{equation*}
\sum_{[x]\in  \mathcal{X}/\mathfrak{S}_{N-M}  }\left(  \sum_{ r:\, \alpha_r\in \mathrm{Supp}([x])}  b_{r}
n_{r}([x]) \right)\left(\sum_{\substack{x\in [x]} }|a_{(I_{0}; x)}|^{2}\right )
=
\sum_{[x]\in  \mathcal{X}/\mathfrak{S}_{N-M}  }\left(  \sum_{ r=1}^{D}  b_{r}
n_{r}([x]) \right)\left(\sum_{\substack{x\in [x]} }|a_{(I_{0}; x)}|^{2}\right )
\,.
\end{equation*}
This is vanishing due to the equation $\sum_{r=1}^{D}b_{r}n_{r}([x])=0$ for any $[x]\in  \mathcal{X}/\mathfrak{S}_{N-M} $, as proved in Lemma \ref{lemmacombinatoricsofpartition}.

\end{proof}

The above results exhibit only part of the combinatorial properties in partial trace.
The actual combinatorial structure in partial trace is much richer.
For example, the quantity considered in \eqref{eqnsummationoverrexpanded}
is closely related to
 \begin{equation}
\mathrm{Tr}_{V_{*}\otimes V_{\Lambda}}\, (\psi\otimes \psi^{*})\,,\quad
\Lambda\cup \{*\}=\{M+1,M+2,\cdots N\}\,,
\end{equation}
whose $(I_{0},I_{0})$-diagonal entry is
 \begin{equation}
\sum_{i_*=1}^{D}
\sum_{e_{i_*}\in V_{*}}
\rho_{(I_{0};i_* )}^{\{1,2,\cdots , M, *\} }\,.
\end{equation}
In particular, if we take $M=1$, then
$\mathrm{Tr}_{V_{*}\otimes V_{\Lambda}}$ defines
an element in $\mathrm{End} (V_{1})$ and we have
 \begin{equation}\label{eqntraceoftrace}
\sum_{*\in  \{2,3,\cdots, N\} }
\sum_{i_*=1}^{D}
\sum_{e_{i_*}\in V_{*}}
\rho_{(I_{0};i_* )}^{\{1, *\} }
=(N-1)\cdot \left( \mathrm{Tr}_{V_{*}\otimes V_{\Lambda}}\, (\psi\otimes \psi^{*})\right)_{(I_{0},I_{0})}\,。
\end{equation}
The summation of the above over $I_{0}$ gives the further trace over $V_{1}$. We therefore have
\begin{equation}
\sum_{I_{0}}
\sum_{*\in  \{2,3,\cdots, N\} }
\sum_{i_*=1}^{D}
\sum_{e_{i_*}\in V_{*}}
\rho_{(I_{0};r )}^{\{1, *\} }
=(N-1) \cdot \mathrm{Tr}_{V} \, (\psi\otimes \psi^{*})\,.
\end{equation}
When combined with Theorem \ref{thmeqnforpartialtraces}, this
will be useful in the  applications discussed in Section \ref{sec:app} below
where we shall prove that the converse statement is also true.

\subsection{The relation between different $M'<M$}
\label{sec:stru}

For different choices of $V_{\Lambda}$, the patterns in the combinatorics of the partial trace shown in Lemma \ref{lemmacombinatoricsofpartition}
in fact only depends on the cardinality $N-M-1$ of $\Lambda$.

For different values of $M$, Theorem \ref{thmeqnforpartialtraces} gives  different sets of relations.
We shall show in this section that the most informative one is
the one with largest possible $M$ subject to the condition $M+1\leq [{N\over 2}]$ (the least possible number of components being traced out according to \eqref{eqnqualitativeanalysis}),
the others are its consequences.\\

Fixing $M$,
consider another value $M'$ such that $M'<M$.
Our argument is by induction. Hence we shall
assume for now that $M'=M-1$.
We single out the component in $M-M'$.
Assume it is the first component, by permutation or relabeling if necessary.\\

Recall that the relations in
Lemma \ref{lemmacombinatoricsofpartition}
is about the combinatorics of $\mathcal{X}/\mathfrak{S}_{N-M}$.
We now show that the solution $\{b_{r}\}$ given in \eqref{eqnsoltobr}
implies the solution $\{b_{r}'\}$.
We start with the fact that each of the partitions $[x]$ satisfies
\begin{equation}
\sum b_{r}n_{r}([x])=0\,,
\quad
\sum n_{r}([x])=N-M\,,
\quad
\sum \alpha_{r} n_{r}([x])=S\,.
\end{equation}
Here the existence of $\{b_{r}\}$
is guaranteed  by induction hypothesis. Later we shall see that the solution given in \eqref{eqnsoltobr} is a natural one consistent with the induction procedure.

The goal is then to prove the existence of $\{b_{r}'\}$
such that the $'$-version of the above equations are satisfied, for any $[x']$.\\

Choose a value $s$ for the first component in the process of taking partial trace
over the $N-M'$ components.
We can then classify $x'$ into two sets: one involves $s$ and the other one does not.
For those not involving the specified $s$, it must involve some other value $\tilde{s}$.
Then we apply the following same reasoning to $[x']=[x]+\tilde{s}$.

If we can prove the result for any possible value of $s$, then by exhausting all the possible values for $s$, we are done with the checking for any $[x']$, as
any $x'$ must be of the form $[x']=[x]+s$ for some $s$.\\

Hence it suffices to consider those involving any fixed value $s$, for which we have $[x']=[x]+s, [x]\in \mathcal{ X}/\mathfrak{S}_{N-M}$, with
\begin{equation}
\sum n_{r}([x'])=N-M'=N-M+1\,,
\quad
\sum  \alpha_{r} n_{r}([x'])=S'\,.
\end{equation}
We set
\begin{equation}\label{eqnbshift}
b_{r}'=b_{r}+\delta_r\,,
\end{equation}
for some $\delta_r$.
We want it to depend only on the
numbers $S\leq S'$ being partitioned and $M=M'+1$ so that we can proceed by induction.\\

Now we compute
\begin{eqnarray*}
\sum_{r=1}^{D} b_{r}' n_{r}([x'])
&=&\sum_{r=1}^{D} b_{r} n_{r}([x])
+\sum_{r=1}^{D} \delta_r n_{r}([x])
+b_{s}+\delta_r\,,\\
&=&(N-M)\delta_r
+b_{s}+\delta_r\\
&=&b_{s}
+(N-M+1)\delta_r \\
&=&b_{s}
+ (N-M')\delta_r\,.
\end{eqnarray*}
We set
\begin{equation}\label{eqnpatternforbr}
b_{r}= \alpha_{r}-{S\over N-M}\,,\quad
 \forall\, r=1,2,\cdots, D\,,
\end{equation}
and
\begin{equation}
\delta_r={S\over N-M}-{S'\over N-M'}\,, \quad \forall\, r=1,2,\cdots, D\,.
\end{equation}
This then
does the job $\sum_{r=1}^{D} b_{r}' n_{r}([x'])=0$.
In fact, from this one can see that $\delta_{r}$ is independent of $r$.
Furthermore, one has from the above and  \eqref{eqnbshift} that
\begin{equation}
b_{r}'= \alpha_{r}-{S'\over N-M'}\,.
\end{equation}
Hence it keeps the pattern for $b_{r}$ shown in \eqref{eqnpatternforbr} unchanged.
Therefore, one can proceed by induction.

\section{Applications}
\label{sec:app}

Theorem~\ref{thmeqnforpartialtraces} is a strong structural condition on the partial trace
$\mathrm{Tr}_{V_{\Lambda}}( \psi\otimes \psi^{*} )\in \mathrm{End}(V_{\Lambda^c})$.
One immediate application is for the overlapping quantum marginal problem when restricting to the
constant weight subspace. For overlapping quantum marginal problem, very few results were known~\cite{chen2014symmetric,chen2016detecting} and most of them can only be applied to small systems. To the best of our knowledge, no further conditions are known if we restrict the pre-image to lie in a given subspace.  \\

To make things precise, we first give the definitions of density operator and density matrix, which are the practical notions in talking about distributions in probability theory.
\begin{definition}[Density operator and density matrix]
Suppose $E$ is a Hermitian vector space.
A density operator $\varrho$
is an element in $\mathrm{End}(E)$
satisfying
\begin{itemize}
\item It is normalized in the sense that $\mathrm{Tr} \,\varrho=1$.
\item It is a self-adjoint, positive definite operator.
\end{itemize}
Fixing an orthonormal basis $\{f_{L}\}$ for $E$, then the density operator $\varrho$ is represented by a matrix $(\varrho_{LL' })$ called the density matrix.
The self-adjoint property translates into the property that the density matrix is Hermitian.
We denote its diagonal entries
by
\begin{equation}
\rho_{L}:=\varrho_{LL}\,.
\end{equation}
\end{definition}

For example, for any unit norm vector $v\in E$, the operator
$v\otimes  v^{*}$
gives a density operator.

\subsection{Quantum marginal problem}

The quantum marginal problem is formulated in the following way.
Consider the Hermitian vector space
$V=\bigotimes_{k=1}^{N} V_{k}$.
For each subset $\{i,j\}\subseteq \{1,2,\cdots, N\}$, define
\begin{equation}
\Lambda_{ij}^{c}:=\{i,j\}\,,
\quad
\Lambda_{ij}:=
\{1,2,\cdots, N\}-\{i,j\}\,.
\end{equation}
Given a collection of density operators $\{\varrho^{\Lambda^{c}_{ij}}\}$, consisting of one density operator(called two-body below) $\varrho^{\Lambda^{c}_{ij}}$ for each $\Lambda^{c}_{ij}$, we want to ask whether there exists a density operator $\varrho^{\{1,2,\cdots, N\}}=\psi\otimes \psi^{*}$ on $V$, supported on the subspace $V_{(w)}$ of constant weight $w$, such that its partial trace over $V_{\Lambda_{ij}}$ satisfies the following relation
\begin{equation}\label{eqnquantummarginal}
\mathrm{Tr}_{V_{\Lambda_{ij}}}\,\varrho^{\{1,2,\cdots, N\}}=\varrho^{\Lambda^{c}_{ij}}\,.
\quad
\forall\, \{i,j\}\subseteq \{1,2,\cdots, N\}\,,
\end{equation}

When there exists such a $\varrho^{\{1,2,\cdots, N\}}$, then $\psi$ is a state in the constant weight subspace $V_{(w)}$.

If not, then there could be two possibilities:

1. there does not exist $\varrho^{\{1,2,\cdots, N\}}$ at all, either on the constant weight subspace or not;

2. there exist some global states but none of them is in a constant weight subspace.\\

Our results of Theorem~\ref{thmeqnforpartialtraces} directly give necessary conditions for this problem.
With respect to some given orthonormal basis of $V$ of the form \eqref{eqnbasisforV}, which is induced by those on the components $V_{k}$,
the diagonal entries of $\varrho^{\Lambda^{c}_{ij}}$
are given by $\{\rho^{\Lambda^{c}_{ij}}_{L}\}_{L \in \Theta_D^{|\Lambda^c|}}$.
If \eqref{eqnquantummarginal}
is true, then these diagonal entries $\{\rho^{\Lambda^{c}_{ij}}_{L}\}_{L \in \Theta_D^{|\Lambda^c|}}$
must coincide with the ones
defined in
\eqref{eqndfnrhodiagonal} in Section \ref{sec:main}.
Hence one must have, for the solution $\{b_{r}\}$ given in \eqref{eqnsoltobr} (which in particular depends on $I_{0}$), the following relations provided in
\eqref{eqnforpartialtraces}:
\begin{equation}\label{eqnstrongnecessaryconditionforconstantweight}
\sum_{r=1}^{D} b_{r}
\sum_{\substack{ *\in  \{M+1,M+2,\cdots, N\}  \\
e_r\in V_{*}}}\rho^{\{1,2,\cdots, M,*\}}_{(I_{0};r)}=0\,,
\quad
\forall\, I_{0}\,,
\end{equation}
as well as the relations obtained by moving the index set $\{1,2,\cdots, M\}$ inside $\{1,2,\cdots, N\}$.

For instance, we can take $M=1$, our result then leads to a new necessary condition for the set of the density operators $(\varrho^{\{1,2\}},\varrho^{\{1,3\}},\cdots,\varrho^{\{1,N\}})$ having a lift into a constant weight subspace:
\begin{equation}\label{eqnweaknecessaryconditionforconstantweight}
\sum\limits_{r=1}^D b_r \sum\limits_{{*\in \{2,\cdots,N\}}\atop{e_r\in V_{*}}}\rho_{(I_0:r)}^{\{1,*\}}=0\,,
\quad
\forall\, I_0\,.
\end{equation}
These linear constraints cannot be obtained from the trivial conditions
\begin{equation}\label{eqntrivialrelation}
\mathrm{Tr}_{V_{p}}(\varrho^{\{1,p\}})=\mathrm{Tr}_{V_{q}}(\varrho^{\{1,q\}})\,,
\quad
\forall\, 2\leq p< q\leq N\,.
\end{equation}

\medskip{}

Here is another closely related problem.
Assuming that $\{\varrho^{\Lambda^{c}_{ij}}\}$ indeed descends from a density operator $\varrho^{\{1,2,\cdots N\}}=\psi\otimes \psi^{*}$, we want to know
to what extent we can know the property of $\psi$ (e.g., deviation from being supported on a constant weight subspace) from the condition \eqref{eqnweaknecessaryconditionforconstantweight} and its permutations.

We now show that the necessary condition
\eqref{eqnweaknecessaryconditionforconstantweight} and its permutations
provided by Theorem~\ref{thmeqnforpartialtraces} is actually sufficient, provided that the
above assumption that $\{\varrho^{\Lambda^{c}_{ij}}\}$ descends from a
density operator $\varrho^{\{1,2,\cdots N\}}=\psi\otimes \psi^{*}$ is met.
Note that the condition \eqref{eqnweaknecessaryconditionforconstantweight} and its permutations are much weaker than the set of relations obtained by permuting
\eqref{eqnstrongnecessaryconditionforconstantweight}.\\

To see this, we assume that \eqref{eqnweaknecessaryconditionforconstantweight} and its permutations are met for a set of $\{b_{r}\}$
given by \eqref{eqnsoltobr} for some $w_{0}$,
\begin{equation}
b_{r}=\alpha_{r}-{-\mathrm{weight}(I_{0})+w_{0}\over N-1}\,,
\quad
r=1,2,\cdots, D\,.
\end{equation}
Then we get
\begin{equation}\label{eqnasumtionw0}
\sum\limits_{r=1}^D (\alpha_{r}-{-\mathrm{weight}(I_{0})+w_{0}\over N-1})\sum\limits_{{*\in \{2,\cdots,N\}}\atop{e_r\in V_{*}}}\rho_{(I_0:r)}^{\{1,*\}}=0\,,
\quad
\forall I_0\,.
\end{equation}

We first decompose $\psi$ into a sum of its projections $\psi_{w}$ to the constant weight subspaces $V_{(w)}$
\begin{equation}
\psi=\sum_{w}\sum_{e_I\in V_{(w)}} a_{I}e_{I}
:=\sum_{w} \psi_{w}\,,\quad a_{I}\in \mathbb{C}\,.
\end{equation}
Then since different $V_{(w)}$'s are orthogonal, we have
\begin{equation}
\mathrm{Tr}_{V_{\Lambda_{1*}}} (\psi\otimes \psi^{*}) =
\sum_{w} \mathrm{Tr}_{V_{\Lambda_{1*}}} (\psi_{w}\otimes \psi_{w}^{*})  \,.
\end{equation}
Denote the diagonal matrices of $\mathrm{Tr}_{V_{\Lambda_{1*}}} (\psi_{w}\otimes \psi_{w}^{*})$ by
$\rho_{(I_0:r)}^{\{1,*\}}(w) $.
Then we have for the diagonal entries that
\begin{equation}\label{eqndecompositionintoconstantweights}
\rho_{(I_0:r)}^{\{1,*\}}=\sum_{w} \rho_{(I_0:r)}^{\{1,*\}} (w)\,.
\end{equation}

Applying Theorem~\ref{thmeqnforpartialtraces} to each component $\psi_{w}$,
we get
\begin{equation}
\sum\limits_{r=1}^D (\alpha_{r}-{-\mathrm{weight}(I_{0})+w\over N-1})\sum\limits_{{*\in \{2,\cdots,N\}}\atop{e_r\in V_{*}}}\rho_{(I_0:r)}^{\{1,*\}} (w)=0\,,
\quad
\forall I_0\,.
\end{equation}
This gives
\begin{equation}\label{eqnrelationintoconstantweights}
\sum_{w}\sum\limits_{r=1}^D (\alpha_{r}-{-\mathrm{weight}(I_{0})+w\over N-1})\sum\limits_{{*\in \{2,\cdots,N\}}\atop{e_r\in V_{*}}}\rho_{(I_0:r)}^{\{1,*\}} (w)=0\,,
\quad
\forall I_0\,.
\end{equation}

On the other hand, from \eqref{eqnasumtionw0}, \eqref{eqndecompositionintoconstantweights}, we also have
\begin{equation}
\sum\limits_{r=1}^D (\alpha_{r}-{-\mathrm{weight}(I_{0})+w_{0}\over N-1})\sum_{w} \sum\limits_{{*\in \{2,\cdots,N\}}\atop{e_r\in V_{*}}}\rho_{(I_0:r)}^{\{1,*\}} (w)=0\,,
\quad
\forall I_0\,.
\end{equation}
One can change the order of summation on $r$ and $w$ and get
\begin{equation}\label{eqnrelationintow0}
\sum_{w}\sum\limits_{r=1}^D (\alpha_{r}-{-\mathrm{weight}(I_{0})+w_{0}\over N-1}) \sum\limits_{{*\in \{2,\cdots,N\}}\atop{e_r\in V_{*}}}\rho_{(I_0:r)}^{\{1,*\}} (w)=0\,,
\quad
\forall I_0\,.
\end{equation}

Taking the difference between \eqref{eqnrelationintoconstantweights} and \eqref{eqnrelationintow0}, we obtain
\begin{equation}
\sum_{w}\sum\limits_{r=1}^D {w-w_{0}\over N-1} \sum\limits_{{*\in \{2,\cdots,N\}}\atop{e_r\in V_{*}}}\rho_{(I_0:r)}^{\{1,*\}} (w)=0\,,
\quad
\forall I_0\,.
\end{equation}
Simplifying this relation a little further, we get
\begin{equation}\label{eqnsimplifedsumbeforesumI0}
\sum_{w}{w-w_{0}\over N-1}  \sum_{*\in \{2,\cdots,N\} }\left(\sum\limits_{i_*=1}^D  \sum\limits_{e_{i_*}\in V_{*}}\rho_{(I_0:i_*)}^{\{1,*\}} (w)\right)=0\,,
\quad
\forall I_0\,.
\end{equation}
The combinatorics in \eqref{eqntraceoftrace} tells that
 \begin{eqnarray}
&\sum_{i_*=1}^{D}
\sum_{e_{i_*}\in V_{*}}
\rho_{(I_{0};i_* )}^{\{1, *\} }(w)
= \left( \mathrm{Tr}_{V_{*}\otimes V_{\Lambda}}\, (\psi_{w}\otimes \psi_{w}^{*})\right)_{(I_{0},I_{0})} \,,\\
&\sum_{*\in \{2,\cdots,N\} }
\sum_{i_*=1}^{D}
\sum_{e_{i_*}\in V_{*}}
\rho_{(I_{0};i_* )}^{\{1, *\} }(w)
=(N-1) \left( \mathrm{Tr}_{V_{\{2,3,\cdots, N\}}}\, (\psi_{w}\otimes \psi_{w}^{*})\right)_{(I_{0},I_{0})} \,.
\end{eqnarray}
Then it follows from \eqref{eqnsimplifedsumbeforesumI0}
that
\begin{equation}\label{eqnexpectationvaluecondition}
\sum_{w}(w-w_{0}) \left( \mathrm{Tr}_{V_{\{2,3,\cdots, N\}}}\, (\psi_{w}\otimes \psi_{w}^{*})\right)_{(I_{0},I_{0})} =0\,,
\quad
\forall I_0\,.
\end{equation}

In particular, a consequence of this says that
the expectation value of the operator $\mathbf{H}=\sum_{k}H_k$ on the density operator $\varrho^{\{1,2,\cdots,N\}}$
is the same as that of the constant operator $w_{0}$. Here $H_k=\mathrm{Id}\otimes\cdots \otimes H\otimes \cdots \mathrm{Id}$ acts nontrivially on $V_k$ only, hence $\mathbf{H}$ must be diagonal since $H$ is in the Cartan subalgebra.
\begin{eqnarray} \label{eqnweakerconsequence}
0&=&\sum_{I_{0}}\sum_{w}(w-w_{0}) \left( \mathrm{Tr}_{V_{\{2,3,\cdots, N\}}}\, (\psi_{w}\otimes \psi_{w}^{*})\right)_{(I_{0},I_{0})} \nonumber\\
&=&\sum_{I_{0}} \sum_{w}
w\left( \mathrm{Tr}_{V_{\{2,3,\cdots, N\}}}\, (\psi_{w}\otimes \psi_{w}^{*})\right)_{(I_{0},I_{0})}
-\sum_{I_{0}} \sum_{w}
w_{0}\left( \mathrm{Tr}_{V_{\{2,3,\cdots, N\}}}\, (\psi_{w}\otimes \psi_{w}^{*})\right)_{(I_{0},I_{0})}  \nonumber \\
&=&\sum_{w}w\sum_{I_{0}}
\left( \mathrm{Tr}_{V_{\{2,3,\cdots, N\}}}\, (\psi_{w}\otimes \psi_{w}^{*})\right)_{(I_{0},I_{0})}
-w_{0} \sum_{w}\sum_{I_{0}}
\left( \mathrm{Tr}_{V_{\{2,3,\cdots, N\}}}\, (\psi_{w}\otimes \psi_{w}^{*})\right)_{(I_{0},I_{0})}  \nonumber \\
&=&\sum_{w}w\cdot \mathrm{Tr}_{V}\, (\psi_{w}\otimes \psi_{w}^{*})
-w_{0} \sum_{w} \mathrm{Tr}_{V}\, (\psi_{w}\otimes \psi_{w}^{*})  \nonumber \\
&=&
\mathrm{Tr}_{V} (\mathbf{H} \varrho^{\{1,2,\cdots, N\}} )
-w_{0}\cdot  \mathrm{Tr}_{V } (   \varrho^{\{1,2,\cdots, N\}} ) \,.
\end{eqnarray}

We remark that \eqref{eqnexpectationvaluecondition} is in fact stronger than this.
From \eqref{eqnentriesofpartialtrace} we can see that
$\mathrm{Tr}_{V_{\{2,3,\cdots, N\}} }\, (\psi_{w}\otimes \psi_{w}^{*})$
only have diagonal terms: this is special as the leftover after the partial trace has only one component.
Therefore, \eqref{eqnexpectationvaluecondition}
actually means that the following two are equivalent operators (contrasting \eqref{eqntrivialrelation})
\begin{equation} \label{eqnstrongerconsequence}
\sum_{w}
 w\left( \mathrm{Tr}_{V_{\{2,3,\cdots, N\}}}\, (\psi_{w}\otimes \psi_{w}^{*})\right)
=\sum_{w}  w_{0}\left( \mathrm{Tr}_{V_{\{2,3,\cdots, N\}}}\, (\psi_{w}\otimes \psi_{w}^{*})\right) \,.
\end{equation}
Taking further the partial trace over $V_{1}$ then yields \eqref{eqnweakerconsequence}.\\

We have shown in \eqref{eqnstrongerconsequence}
that
\begin{equation}
 \mathrm{Tr}_{V_{\{2,3,\cdots, N\}}}\,
(\mathbf{H}-H_{1})\psi\otimes \psi^{*}
+(H_{1}-w_{0})  \mathrm{Tr}_{V_{\{2,3,\cdots, N\}}}\,\psi\otimes \psi^{*}=0\,.
\end{equation}
Permuting the index from $1$ to $k$ gives
\begin{equation}
 \mathrm{Tr}_{V_{\{1,2,3,\cdots, N\}-\{k\}}}\,
(\mathbf{H}-H_{k})\psi\otimes \psi^{*}
+(H_{k}-w_{0})  \mathrm{Tr}_{V_{\{1,2,3,\cdots, N\}-\{k\}}}\,\psi\otimes \psi^{*}=0\,,
\quad
\forall k=1,2,\cdots N\,.
\end{equation}
Multiplying this by a polynomial operator $f(H_{k})$,
and summing over $k$,
we then get
\begin{equation}
 \mathrm{Tr}_{V}\,   \sum_{k=1}^{N}\left(
(\mathbf{H}-w_{0}) f(H_{k})\right)
\psi\otimes \psi^{*}=0\,.
\end{equation}

Taking $f(H_{k})=w_{0}$ gives
\begin{equation}
 \mathrm{Tr}_{V}\,
 (N \mathbf{H} w_{0}-N w_{0}^2)
\psi\otimes \psi^{*}=0\,.
\end{equation}
Taking $f(H_{k})=H_{k}$ yields
\begin{equation}
 \mathrm{Tr}_{V}\,
 (\mathbf{H}^2-\mathbf{H} w_{0})
\psi\otimes \psi^{*}=0\,.
\end{equation}
Combining the above two, we get
\begin{equation}
 \mathrm{Tr}_{V}\,
 (\mathbf{H}-w_{0})^2
\psi\otimes \psi^{*}=
\sum_{w}
 \mathrm{Tr}_{V_{(w)  }}
  (w-w_{0})^2
\psi_{w}\otimes \psi^{*}_{w}=0
\,.
\end{equation}
This can be true only when $\psi\otimes \psi^{*}$
is supported on the constant weight subspace
$V_{(w_{0})}$.\\

Hence we have shown that one can determine whether a state is supported on a constant weight subspace by all of its two-body local information.
The proof above also shows that the vanishing of fluctuation of $H$ would give another necessary and sufficient condition to this problem.
However, in the case leakage exists, our method gives a more practical and powerful criteria than the mean and fluctuation method.

\subsection{Perfect tensor}

As another concrete example of our applications, we now use our conditions to study the notion of perfect tensor, which is recently widely studied in the theory of AdS/CFT, and is understood as an interesting proposal to realize the holographic principle in many-body quantum systems. Perfect tensors can build tensor network state exhibiting interesting holographic correspondence \cite{pastawski2015holographic}. In particular, the tensor network made by perfect tensors derives the Ryu-Takayanagi formula of holographic entanglement entropy, namely, the entanglement entropy of
the boundary quantum system equals the minimal surface area in the bulk \cite{qi2013exact,han2017loop}.

Furthermore, recently it has been shown that perfect tensors can represent quantum channels which are of strongest quantum chaos \cite{hosur2016chaos}. The quantum transition defined by perfect tensors turns out to maximally scramble the quantum information such that the initial state cannot be recovered by local measurements. It has also been suggested that a perfect tensor should represent the holographic quantum system dual to the bulk quantum gravity with a black hole.

\begin{definition}[Perfect tensor]\label{dfnperfecttensor}
A vector $\psi\in V$ is called a perfect tensor if for all possible choices $\Lambda$ satisfying $|\Lambda|\geq\frac{N}{2}$, the condition
\begin{equation}
\mathrm{Tr}_{V_{\Lambda}}( \psi\otimes \psi^{*} )=c_{|\Lambda|}\cdot \mathbb{I}_{V_{\Lambda^{c}}}\,
\end{equation}
is satisfied for some non-vanishing constant $c_{|\Lambda|}$.
\end{definition}

The following result follows from Theorem \ref{thmeqnforpartialtraces}.

\begin{theorem}\label{thmeqnforperinvtensor}
Fixing $N\geq 4$,
then for any $w$, there does not exist a perfect tensor in the constant weight subspace with weight $w$.
\end{theorem}
\begin{proof}
We prove by contradiction.
Suppose there exists a perfect tensor $\psi$.
We can then take $\Lambda^{c}$ with cardinality $M+1$ such that the condition
in \eqref{eqnqualitativeanalysis} is fullfilled.
That is,
\begin{equation}\label{eqnrequirementfornonemptynessofI0}
M+1\leq N-[{N+1\over 2}]=[{N\over 2}]\,.
\end{equation}
Then according to Definition \ref{dfnperfecttensor}, one must have
\begin{equation}
\mathrm{Tr}_{V_{\Lambda}}( \psi\otimes \psi^{*} )=c_{|\Lambda|}\cdot \mathbb{I}_{V_{\Lambda^{c}}}\,,
\end{equation}
for some non-vanishing constant $c_{|\Lambda|}$.
It is easy to see that $c_{|\Lambda|}$ only depends on the cardinality of $\Lambda$: the further trace over $V_{\Lambda^{c}}$
should give a multiple of the identity endomorphism which is independent of
the choice of $\Lambda$.

We now consider the entries $ \rho_{(I_{0};r )}^{{\Lambda^{c}}}$ constructed in
\eqref{eqndiagonalpieces}.
All of them are equal to $c_{|\Lambda|}$ which without loss of generality
can be normalized to $1$.
Then we have

 \begin{equation}
\rho_{(I_{0};r )}^{{\Lambda^{c}}}
=
1 \,,
\quad
\forall\, (I_{0};r)\in \Theta_D^{\Lambda^{c}}\,.
\end{equation}

We now show that if $M\neq 0$, that is, the set $I_{0}$ is nonempty, then there always exists $I_0$ such that
\begin{equation}\label{eqnconditionforcontradiction}
\sum_{r=1}^{D} b_{r}
\neq 0\,.
\end{equation}
The condition
$M\geq 1$ requires $N\geq 4$ according to \eqref{eqnrequirementfornonemptynessofI0}.
To check the condition \eqref{eqnconditionforcontradiction}, we compute
\begin{equation}
\sum_{r=1}^{D}b_{r}=\sum_{r=1}^{D} \alpha_{r}- D\cdot {S\over N-M}\,.
\end{equation}
Due to the structure theory of representations, one has $\sum_{r=1}^{D} \alpha_{r}=0$.
Hence the condition boils down to
\begin{equation}
S=-\mathrm{weight}(I_{0})+w\neq 0\,.
\end{equation}
This can always be satisfied by choosing a suitable $I_{0}$,
which is contradictory with the claim in Theorem \ref{thmeqnforpartialtraces}.
Hence there does not exist such a perfect tensor.
\end{proof}
\medskip{}

Perfect tensors also have an intimate connection to quantum error-correcting codes~\cite{raissi2018optimal}. An $N$-spin perfect tensor can be equivalently viewed as a length-$N$ quantum error-correcting code encoding a single quantum state,  with the code
distance $\delta=\lfloor N/2\rfloor+1$. Our results hence indicates in the constant weight subspace, there is no such
code exist.

We will now use our results to further understand the existence of invariant perfect tensors. Invariant tensors are the tensors in $V$ with vanishing total angular momentum. They play an important role in the theory of loop quantum gravity \cite{han2007fundamental,rovelli2014covariant}, and particularly the structure of Spin-Networks \cite{Penrose,rovelli1995spin}. Spin-network states, as quantum states of gravity, are networks of invariant tensors, and represent the quantization of geometry at the Planck scale. Classically an arbitrary three-dimensional geometry can be discretized and built piece by piece by gluing polyhedral geometries. The spin-network state built by invariant tensors quantizes the geometry made by polyhedra. As the building block of spin-network, the $n$-valent invariant tensor represents the quantum geometry of a polyhedron with $n$ faces. The reason in brief is that the quantum constraint equation $\sum_{i=1}^n{J}_i\psi=0$ (vanishing total angular momentum) is a quantum analog of the polyhedron closure condition $\sum_{i=1}^n\vec{A}_i=0$ in three-dimensional space (see e.g. Appendix A in \cite{li2017invariant} for details).

Given that both invariant tensors and perfect tensors relate to quantum gravity from different perspectives, it is interesting to incorporate the idea of perfect tensors with that of invariant tensors, and define a new concept that we call invariant perfect tensor.

\begin{definition}[Invariant perfect tensor]\label{dfnperfectinvtensor}
A nonzero vector $\psi\in V$ is called an invariant tensor if
\begin{equation}
H\psi=0,\ X\psi=0,\ Y\psi=0\,.
\end{equation}
A nonzero vector $\psi\in V$ is called an invariant perfect tensor if it is both perfect and invariant.
\end{definition}

A partial study of invariant perfect tensors has been carried out in \cite{li2017invariant}, which shows that at $N=2,3$ invariant tensors are always perfect, but strictly there is no invariant perfect tensor at $N=4$, although invariant tensors generically approximate perfect tensors asymptotically in large $j$.

The result in Theorem \ref{thmeqnforperinvtensor} generalizes the conclusion for invariant perfect tensor to arbitrary $N\geq4$. Since invariant tensors live in the constant weight $w=0$ subspace, we obtain the following.

\begin{corollary}
There does not exist invariant perfect tensor for any $j$, for $N\geq 4$.
\end{corollary}

$N=3$ invariant tensors are employed in spin-network states for 2+1 dimensional gravity, while $N\geq4$ invariant tensors build spin-network states for 3+1 dimensional gravity~\cite{perez2003spin,baez1998spin}.
The above results show that the entanglement exhibited by the local building block of quantum gravity (at Planck scale) is not as much as a perfect tensor. So the holographic property of quantum states is obscure at the Planck scale. The holography displayed by quantum gravity at semi-classical level then suggests that in order to understand quantum gravity using tensor networks, the entanglement of perfect tensor, as being important to understand holography, should be a large scale effect coming from coarse-graining the Planck scale microstates. Namely although the perfect tensor is missing at the Planck scale, it may emerge approximately at the larger scale, and makes tensor networks demonstrate holography. This idea is very much consistent with the recent proposal in \cite{han2017loop}, which shows the spin-network states in 3+1 dimensions can indeed give tensor networks exhibiting holographic duality at the larger scale. Then it is interesting to understand how (approximate) perfect tensors emerge from non-perfect invariant tensors via coarse-graining from the Planck scale to larger scale. The research in this perspective will be reported in a future publication.

\section{Generalizations and extensions}

\subsection{Generalization to $\mathrm{SU}(n)$}
\label{sec:sun}

We have considered in the above the case where
 $V$ is the tensor product of $N$ copies of a not necessarily irreducible representation $W$ of $\mathrm{SU}(2)$.
We now generalize this to the $\mathrm{SU}(n)$ case.\\

Consider $V=\bigotimes_{k=1}^{N}V_{k}$, where all of the $V_{k}$'s
are given by
the same irreducible representation $W$ of $\mathrm{SU}(n)$ of dimension $D$.
Suppose a basis of the Cartan subalgebra of $\mathrm{SU}(n)$ is given by $H^{(1)},H^{(2)},\cdots, H^{(n-1)}$.
One then has the weight space decomposition
\begin{equation}\label{eqnSUndecomposition}
W=\bigoplus_{\vec{\alpha}\in \Delta} W_{\vec{\alpha}}\,,
\end{equation}
where $\Delta$
is the weight space and $W_{\vec{\alpha} }$ is the eigenspace with the weight vector $\vec{\alpha}=
(\alpha^{(1)},\alpha^{(2)},\cdots \alpha^{(n-1)}) $, that is
\begin{equation}
 W_{\vec{\alpha} }=\{v\in W~| H^{(i)} v=\alpha^{(i)} v\,,\forall i=1,2,\cdots, n-1\}\,.
\end{equation}
In particular, the action of any element
in the Cartan subalgebra, symbolically denoted by $H^{(*)}$, is diagonal on $W_{\vec{\alpha}}$
and hence on $W$.\\

The above decomposition is orthogonal.
We choose an orthonormal basis $e_{1},e_{2},\cdots, e_{D}$ whose eigenvalues under
$H^{(*)}$ are given by
$\vec{\alpha}_{1}^{(\star)},\vec{\alpha}_{2}^{(\star)},\cdots, \vec{\alpha}_{D}^{(\star)}$.

The constant-weight condition becomes the condition
that under the action of the Cartan subalgebra generated by $H^{(1)},H^{(2)},\cdots, H^{(n-1)}$, the weight vector is constant, say
$\vec{w}=(w^{(1)},w^{(2)},\cdots, w^{(n-1)})$.
In particular, the weight under $H^{(*)}$ is the fixed number $w^{(*)}$.

Then everything discussed in the $\mathrm{SU}(2)$ case follows.
The same reasoning also works when $W$ is not irreducible, in which case a similar orthogonal decomposition in
\eqref{eqnSUndecomposition} still exists, thanks to the structure theory for finite dimensional representations of the Lie group $\mathrm{SU}(n)$.

This then shows that there is no perfect tensor in a constant weight subspace
for the group $G=\mathrm{SU}(n)$ when $N\geq 4$.

\subsection{Relaxing the constant weight subspace condition}
\label{sec:gen}

We now discuss to what extent one can relax the constant weight condition.

Recall that the combinatorics in partial trace allows one to pass from the space $\mathcal{X}$ to its quotient
$\mathcal{X}/\mathfrak{S}_{N-M}$.
What makes the proof in
Theorem \ref{thmeqnforperinvtensor} work is the relation
\eqref{eqncombinatoricsofpartition} in Lemma
\ref{lemmacombinatoricsofpartition}
\begin{equation}
\sum_{r=1}^{D}b_{r}n_{r}([x])=0\,,\quad \forall\, [x]\in \mathcal{X}/\mathfrak{S}_{N-M}\,,
\end{equation}
with the condition in \eqref{eqnconditionforcontradiction}
\begin{equation}
\sum_{r=1}^{D} b_{r}
\neq 0\,.
\end{equation}
The condition
\begin{equation}
\sum_{r=1}^{D} n_{r}([x])=N-M\,,\quad \forall\, [x]\,\in \mathcal{X}/\mathfrak{S}_{N-M}
\end{equation}
 is automatically satisfied, according to \eqref{eqnconstantweightcombinatorics} which follows from the definition of $\mathcal{X}$.\\

Suppose we impose a certain constraint which is not necessarily the constant weight condition.
Assume that the
set of vectors satisfying this constraint, required to be independent of the ordering of the $N$ components, is indexed by the set $\mathcal{Y}$.
Denote the cardinality of the quotient $\mathcal{Y}/\mathfrak{S}_{N-M}$ by $P$.
Then the non-existence of perfect tensors in the space $\mathcal{Y}$ would follow if the following conditions are satisfied
\begin{equation}\label{eqngeneralizingconstraint}
\sum_{r=1}^{D}b_{r}n_{r}([y])=0\,,\quad
\sum_{r=1}^{D} b_{r}
\neq 0\,, \quad \forall\,[y]\in \mathcal{Y}/\mathfrak{S}_{N-M}\,.
\end{equation}
We fix a set of representatives $\{[y_{i}], i=1,2,\cdots, P\}$ for $\mathcal{Y}/\mathfrak{S}_{N-M}$, and denote
the matrix of frequencies by
\begin{equation}
A=(A_{ir})=(n_{r}([y_{i}]))_{i=1,2\cdots P;r=1,2\cdots D}\,.
\end{equation}
Then the above two equations become the conditions for the vector $b=(b_{1},\cdots, b_{D})^{t}$
\begin{equation}
Ab=0\,,
\quad
(1,1,\cdots, 1) b\neq 0\,.
\end{equation}
We denote $\tilde{A}$ to be the matrix obtained by adding a row of $1$'s below the $P$-th row of $A$.
The existence of such a vector $b$ is equivalent to the condition that
\begin{equation}\label{eqnrankconditions}
\mathrm{rank}\, A<D\,,
\quad  \mathrm{rank}\, \tilde{A}= \mathrm{rank}\, A+1\,.
\end{equation}

\begin{example}
Consider the case where each $V_{k}$ in $V=\bigotimes_{k=1}^{N} V_{k}$ is the irreducible representation $\mathrm{SU}(2)$ of dimension $D=2j+1$.
We put the constant weight condition.
This is the main interest in this work.
In this case, it is straightforward to show that the cardinality $P$ of $\mathcal{X}/\mathfrak{S}_{N-M}$ is
\begin{equation}
P=\mathrm{Coeff}_{S+D(N-M)}\prod_{r=1}^{D} {1\over 1-t_{r}}\,,
\end{equation}
where $t_{r}, r=1,2,\cdots D$ are formal parameters with grading $2r-1$ respectively,
and $\mathrm{Coeff}_{k}f(t_1,t_2,\cdots,t_D)$ is the sum of the
coefficients
of the grading-$k$ terms in the Taylor expansion of $f(t_1,t_2,\cdots,t_D)$ near $(t_1,t_2,\cdots,t_D)=(0,0,\cdots,0)$.
Lemma \ref{lemmacombinatoricsofpartition} applies to this case.

The existence of a non-trivial solutions
tells that $\mathrm{rank}\, A \leq D-1$.
In fact, more is known and we can show that
$\mathrm{rank}\, A =D-1$ for generic $D$ and $N-M$.
For simplicity we consider the case where each $V_{k}$ in $V=\bigotimes_{k=1}^{N} V_{k}$ is the irreducible representation of $\mathrm{SU}(2)$ with dimension $D=2j+1$.
The general case where $V_{k}$ is not irreducible is similar.

Now according to \eqref{eqnconstantweightcombinatorics}, we see that if $[x]$
is a certain partition in $\mathcal{X}/\mathfrak{S}_{N-M}$, then the following is also a partition if $N-M\geq 2$
\begin{equation}
[x]+  \alpha_{a}\cdot 1+ \alpha_{b}\cdot 1\,,
\end{equation}
where $\alpha_{a}, \alpha_{b}$ are subject to the condition that $\alpha_{a}+ \alpha_{b}=0$.
Furthermore, one also has the combinations
\begin{equation}
[x]+   \alpha_{a}  \cdot 2+\alpha_{b-1} \cdot 1 + \alpha_{b+1} \cdot 1 \,,\quad
[x]+ \alpha_{a-1} \cdot 1 +\alpha_{a+1}\cdot 1 +  \alpha_{b}\cdot 2 \,.
\end{equation}
The former gives a genuine partition provided the relation $1\leq b-1, b+1\leq D$ can be satisfied, which is the case
when $D\geq 3, N-M\geq 3$.
The latter is similar.
We shall refer to the case $D\geq 3, N-M\geq 3$ the generic case, the others are isolated cases which can be dealt with easily.

We then take the frequency vectors corresponding to the partitions in the set
\begin{equation}
[x]\,,
\quad [x]+ \alpha_{a}\cdot 1 + \alpha_{b} \cdot 1 \,,
\quad
[x]+ \alpha_{a}\cdot 2 + \alpha_{b-1} \cdot 1  + \alpha_{b+1} \cdot1   \,, \quad b\leq D-1
\,,
\quad
[x]+\alpha_{1} \cdot 1  +\alpha_{3} \cdot 1  +\alpha_{D-1}\cdot 2  \,.
\end{equation}
By computing the determinants inductively, it is easy to see that the corresponding frequency vectors
span a vector space of dimension at least $D-1$.
Combing the relation $\mathrm{rank}\, A \leq D-1$ implied by Lemma \ref{lemmacombinatoricsofpartition}, we are then led to the conclusion that
$\mathrm{rank}\, A = D-1$ for generic $D, N-M$.
\end{example}

\begin{example}
As another example, we put the constraint
\begin{equation}
\sum_{r=1}^{D}\alpha_{r}^2 n_{r}=S\,,
\end{equation}
where $\alpha_{r},r=1,2,\cdots, D$, are the weights of the basis $\mathcal{B}=\{e_{1},e_{2},\cdots, e_{D}\}$.
The solution $\{b_{r}\}_{r=1}^{D}$ to \eqref{eqngeneralizingconstraint}
then exists if we choose $I_{0}, S$ suitably.
\end{example}

\section*{Acknowledgments}

We thank Markus Grassl for helpful discussions. J.~C. was supported by the Department of Defence. M. H. acknowledges support from the US National Science Foundation through grant PHY-1602867, and Start-up Grant at Florida Atlantic University, USA. Y.~L. acknowledges support from Chinese Ministry of Education under grants No.20173080024. B.~Z. is supported by NSERC and CIFAR. J.~Z. was supported in part by Perimeter Institute for Theoretical Physics. Research at Perimeter Institute is supported by the Government of Canada through Innovation, Science and Economic Development Canada and by the Province of Ontario through the Ministry of Research, Innovation and Science.

\end{document}